\documentclass[sigconf]{acmart}
\usepackage{graphicx}%
\usepackage{multirow}%
\usepackage{mathrsfs}%
\usepackage[ruled,vlined]{algorithm2e}
\usepackage[pdf]{graphviz}
\usepackage[table]{xcolor}
\usepackage[title]{appendix}%
\usepackage{booktabs}%
\usepackage{tikz}
\graphicspath{{images/}}
\usepackage[pdf]{graphviz}
\usepackage{subcaption}
\usepackage{natbib}

\usepackage{graphicx}%
\usepackage{listings}%
\usepackage{manyfoot}%
\usepackage{mathrsfs}%
\usepackage{minted}%
\usepackage{multirow}%
\usepackage{textcomp}%
\usepackage{amsmath}

\usepackage{tcolorbox}
\tcbuselibrary{skins}
\newtcbox{\inlinecode}{on line, boxrule=0pt, boxsep=0pt, arc=0pt, 
  colback=rowgray, colframe=rowgray, top=1pt, bottom=1pt, 
  left=2pt, right=2pt, fontupper=\ttfamily\small}

\DontPrintSemicolon

\newcommand{\Desc}[2]{\makebox[2em][l]{#1:}#2}
\SetKwInOut{Input}{Input}
\SetKwInOut{Output}{Output}
\SetKwFunction{Collapse}{CollapseMultigraph}
\SetKwFunction{DFSuntil}{DFS‑Until}
\SetKwFunction{Enqueue}{Enqueue}
\SetKwFunction{Dequeue}{Dequeue}
\SetKwProg{Fn}{Function}{}{}
\SetKwComment{Comment}{$\triangleright$\ }{}

\definecolor{steelblue}{RGB}{70,130,180}
\definecolor{rowgray}{RGB}{245,245,245}
\definecolor{plblue}{RGB}{70,130,180}
\definecolor{plbluelight}{RGB}{100,149,237}
\definecolor{tagblue}{RGB}{0,0,210}
\definecolor{darkgray}{RGB}{64,64,64}
\definecolor{algorithmblue}{RGB}{0,0,179}

\usepackage{tikz}%
\usepackage{adjustbox}%
\usepackage{array}%
\newcolumntype{R}[2]{%
    >{\adjustbox{angle=#1,lap=\width-(#2)}\bgroup}%
    l%
    <{\egroup}%
}
\AtBeginDocument{%
    }

\setcopyright{acmlicensed}
\copyrightyear{2018}
\acmYear{2018}
\acmDOI{XXXXXXX.XXXXXXX}
\acmConference[Conference acronym 'XX]{Make sure to enter the correct
    conference title from your rights confirmation email}{June 03--05,
    2018}{Woodstock, NY}
\acmISBN{978-1-4503-XXXX-X/2018/06}

\newcommand{\bigO}{\ensuremath{\mathcal{O}}}

\begin{document}

%\title{Supporting Optional Lines in Block Ordering Problems}
% David: My attempt at a fun title
\title{Choose Your Own Solution: Supporting Optional Blocks in Block Ordering Problems}

\author{Skyler Oakeson}
\affiliation{%
	\institution{Utah State University}
	\city{Logan}
	\state{UT}
	\country{USA}}
\email{skyler.oakeson@usu.edu}

\author{David H. Smith IV}
\affiliation{%
	\institution{Virginia Tech}
	\city{Blacksburg}
	\state{VA}
	\country{USA}}
\email{dhsmith4@vt.edu}

\author{Jaxton Winder}
\affiliation{%
	\institution{Utah State University}
	\city{Logan}
	\state{UT}
	\country{USA}}
\email{jaxton.winder@usu.edu}

\author{Seth Poulsen}
\affiliation{%
	\institution{Utah State University}
	\city{Logan}
	\state{UT}
	\country{USA}}
\email{seth.poulsen@usu.edu}

%\author{Anon}
%\email{anon@anon.com}
%\orcid{1234-5678-9012}
%\affiliation{%
%	\institution{University of Anon}
%	\city{Anon}
%	\state{Anon}
%	\country{Anon}
%}
%
%\author{Anon}
%\email{anon@anon.com}
%\orcid{1234-5678-9012}
%\affiliation{%
%	\institution{University of Anon}
%	\city{Anon}
%	\state{Anon}
%	\country{Anon}
%}
%
%\author{Anon}
%\email{anon@anon.com}
%\orcid{1234-5678-9012}
%\affiliation{%
%	\institution{University of Anon}
%	\city{Anon}
%	\state{Anon}
%	\country{Anon}
%}
%
%\author{Anon}
%\email{anon@anon.com}
%\orcid{1234-5678-9012}
%\affiliation{%
%	\institution{University of Anon}
%	\city{Anon}
%	\state{Anon}
%	\country{Anon}
%}

\begin{abstract} This paper extends the functionality of block ordering
	problems (such as Parsons problems and Proof Blocks) to include optional
	blocks. We detail the algorithms used to implement the optional block
	feature and present usage experiences from instructors who have integrated
	it into their curriculum. The optional blocks feature enables instructors
	to create more complex Parsons problems with multiple correct solutions
	utilizing ommitted or optional blocks. This affords students a method to
	engage with questions that have several valid solutions composed of
	different answer components. Instructors can specify blocks with multiple
	mutually exclusive dependencies, which we represent using a multigraph
	structure. This multigraph is then collapsed into multiple directed acyclic
	graphs (DAGs), allowing us to reuse existing algorithms for grading block
	ordering problems represented as a DAG. We present potential use cases for
	this feature across various domains, including helping students learn Git
	workflows, shell command sequences, mathematical proofs, and Python
	programming concepts. \end{abstract}

\begin{CCSXML}
	<ccs2012>
	<concept>
	<concept_id>10003456.10003457.10003527</concept_id>
	<concept_desc>Social and professional topics~Computing education</concept_desc>
	<concept_significance>500</concept_significance>
	</concept>
	</ccs2012>
\end{CCSXML}

\ccsdesc[500]{Social and professional topics~Computing education}

\keywords{Parsons Problems, CS1, Distractors, Classical Test Theory}

\maketitle

\section{Introduction} Parsons problems provide a structured and less
intimidating entry point for novice programmers compared to a blank page and an
open-ended programming task. This block-ordering technique of test taking
presents learners with pre-written code blocks that must be arranged in the
correct order, helping them focus on control flow and logic without being
overwhelmed by syntax or boilerplate \cite{ericson2022parsons}. Since their
introduction, Parsons problems by Dale Parsons and Patricia Haden
\cite{parsons2006parson} have inspired many different tools and delivery
methods for the probelms and their variants.

\begin{figure}
	\includegraphics[
		width=0.9\columnwidth,
		alt={Screenshot of a Python programming
				block-ordering interface. The problem asks students to construct a function that computes
				the sum of two numbers with multiple correct solutions. The top panel shows draggable
				Python code blocks including function definition, return statements, and different
				approaches to summing variables. The bottom panel shows a partially constructed solution
				with some blocks already placed.}
	]{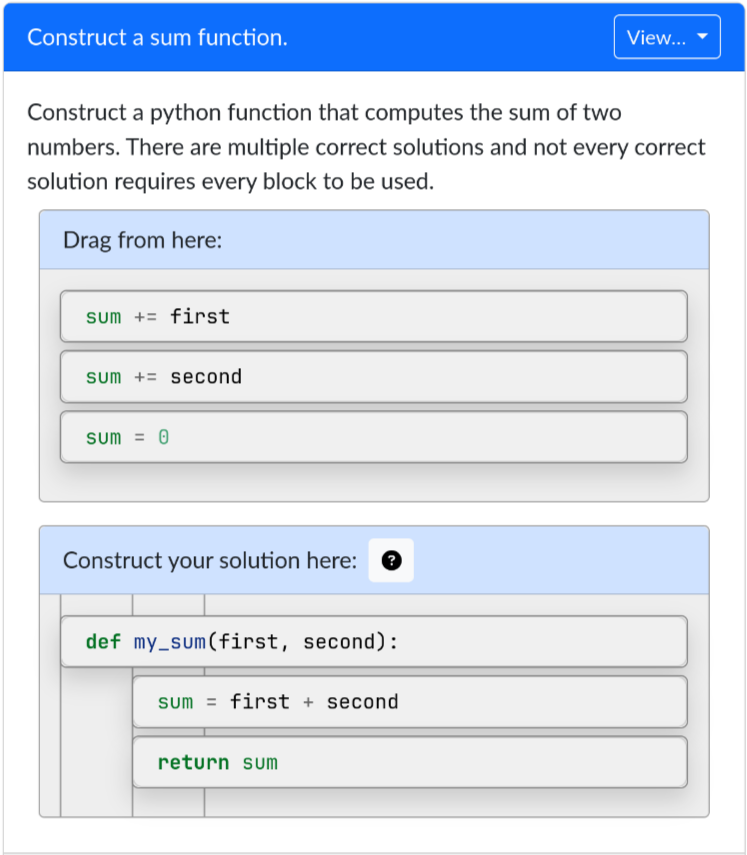}
	\caption{An example of a question utilizing the multigraph collapse algorithm
		in a question that allows multiple ways of returning the sum of two
		values. Between all six blocks presented, there exists four valid solutions, using
		different 		blocks and different orderings of these blocks.}
	\label{fig:sum_pic}

\end{figure}

\citet{poulsen2022proof} expanded the basic functionality of one
correct ordering block-ordering problems to include multiple correct orderings
of the same set of blocks. This was achieved by representing problems as
directed acyclic graphs (DAGs). In this approach, each node represents a block
of the problem, and each edge shows a dependency between two blocks.
Representing the problem as a DAG ensures that the question blocks are reliant
on a previous block and does not allow blocks to depend on themselves.
This design allows instructors to easily create problems with multiple valid
orderings while the underlying systems handles the DAG creation and grading of
student submissions \cite{poulsen2023efficient}. They also extended the idea of
block-ordering problems into mathematical proof construction. Optionally
ordered block-ordering problems have proven valuable across diverse fields,
including networking, chemistry, and programming.

However, a key limitation of the \citet{poulsen2022proof} optionally ordered
block-ordering problems was the inability to include \emph{optional blocks} or
blocks in a question that may be valid in one solution but not in another. As a
result, instructors were constrained to creating problems where every correct
solution had to include a fixed set of blocks that must be included or
excluded, unless they resorted to execution-based feedback~\cite{du2020review}.
However, many instructors prefer line-based feedback because it can give
students more information about how to fix their solutions, does not require
writing test cases, and can be used for proofs in addition to code.

Our goal for this project was to extend the functionality of prior work done by
\citet{poulsen2022proof} and add the possibility for blocks to be substituted
or omitted on top of the already implemented optional ordering feature. To
support this functionality, we introduce a multigraph-based model for
representing dependencies between blocks. This enables us to encode multiple
dependency paths using \emph{colored} edges, which represent alternative
logical progressions. We present an algorithm to collapse a multigraph into all
of its possible valid DAGs (Section~\ref{sec:implementation}), allowing us to
reuse the existing Proof Blocks grading algorithm~\cite{poulsen2023efficient}
to compute the students score. We argue that modeling questions in this way
better reflects the process of real-world problem-solving while further
reinforcing the insight to students that problems can have more than one
equally valid answer.

This paper makes the following contributions:

\begin{itemize}
	\item An interface for instructors to write optional block block-ordering problems;
	\item Use cases of optional block block-ordering problems in domains including introductory
	      programming, discrete math, and teaching command shell; and
	\item Detailing an algorithm that collapses a multigraph into all of its valid
	      \emph{directed acyclic graphs} or DAGs, allowing existing DAG based grading algorithms to be
	      used with problems containing optional lines.
\end{itemize}

\begin{figure}
	\begin{lstlisting}[language=HTML,
    basicstyle=\ttfamily\small,
    keywordstyle=\color{darkgray},
    stringstyle=\color{steelblue},
    commentstyle=\color{darkgray},
    frame=lines,
%    numbers=left,
    numberstyle=\tiny,
    breaklines=true,
    showstringspaces=false,
    morekeywords={correct,tag,depends,final},
    escapeinside={(*}{*)}]
(*\color{algorithmblue}<pl-answer*) tag="A" depends="" indent="0">
    def my_sum(first, second): (*\color{algorithmblue}</pl-answer>*)
(*\color{algorithmblue}<pl-answer*) tag="B" depends="A" indent="1">
    sum = 0 (*\color{algorithmblue}</pl-answer>*)
(*\color{algorithmblue}<pl-answer*) tag="C" depends="B" indent="1">
    sum += first (*\color{algorithmblue}</pl-answer>*)
(*\color{algorithmblue}<pl-answer*) tag="D" depends="B" indent="1">
    sum += second (*\color{algorithmblue}</pl-answer>*)
(*\color{algorithmblue}<pl-answer*) tag="E" depends="A|B" indent="1">
    sum = first + second (*\color{algorithmblue}</pl-answer>*)
(*\color{algorithmblue}<pl-answer*) tag="F" depends="C,D|E" indent="1" final="true">
    return sum (*\color{algorithmblue}</pl-answer>*)
    \end{lstlisting}
	\caption{HTML which generates the problem shown in Figure~\ref{fig:sum_pic}.
		Edges in the dependency graph (which is shown in Figure~\ref{fig:multigraph-example}) are
		given in the \texttt{depends} attribute.
		Multigraph dependencies are specified using the pipe operator
		\texttt{|}. In the given problem, block E can depend on either block A or block B, and block F
		can depend on block C and block D or else just block E.
		More details of the HTML specification are in the
		documentation.\protect\footnotemark}
	\label{fig:sum_html}
\end{figure}

\section{Motivating Example} \footnotetext{Documentation link omitted for
	anonymous review.} As a motivating example, we will look at a Parsons problem
which asks a student to construct a function to sum two numbers. The problem is
shown in Figure~\ref{fig:sum_pic}, and the HTML configuration which generates
the question is shown in Figure~\ref{fig:sum_html}. Students can choose their
own preferred method. The student is able to use the more verbose method of
declaring a variable and then adding to the sum variable over multiple lines,
or they can simply assign \texttt{sum} to the value \texttt{first + second}.
Within these broader choices, there are additional choices the student can make
as well, with a total of $4$ possible solutions, coming from $3$ constructible
dependency DAGs.

While writing such a problem would be possible when giving execution-based
feedback, before the contributions of this paper, writing this problem would
not have been possible when giving line-based feedback. (Definitions of
line-based and execution-based feedback are given by \citet{du2020review}).
Figure~\ref{fig:multigraph-example} shows the process by which we construct the
correct dependency relations to model this particular question.

\begin{figure*}
	\centering
	\includegraphics[width=\textwidth]{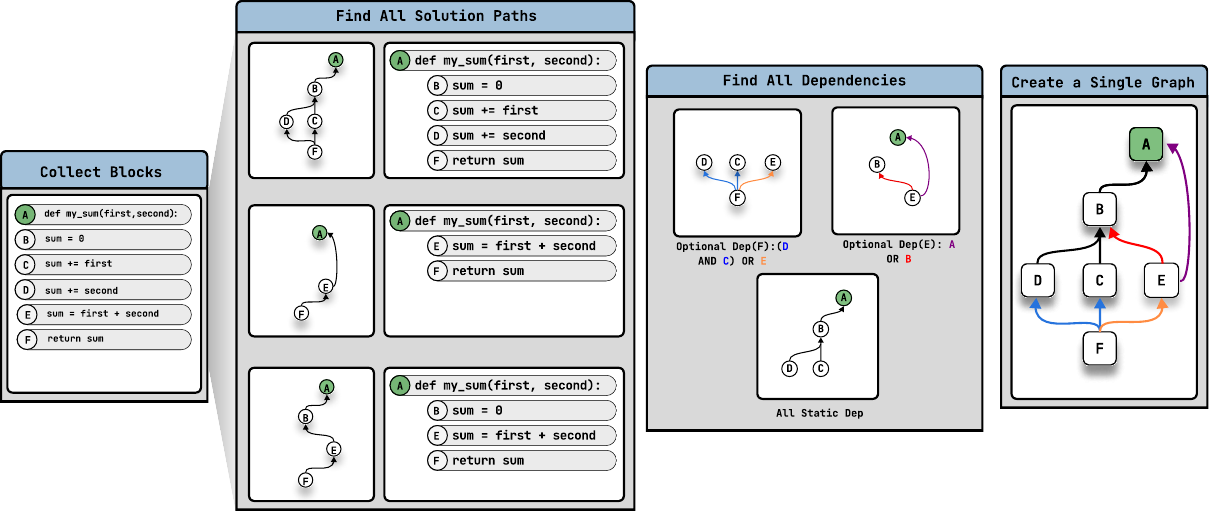}
	\caption{We begin (1) with a collection of nodes that can all be used to
		(2) construct a set of graphs which lead to the source node. We then (3) look at all
		optional dependencies---where a given node can be dependent on different sets of
		nodes---and list those relationships. Finally, we (4) create a single graph with these
		optional dependencies embedded in the edge information. This representation allows the
		multigraph grading algorithm, at a high level, to identify if the source code can be
		reached from a given sink node using the graph we constructed.}
	\label{fig:multigraph-example}
\end{figure*}

\section{Related Work}
\subsection{Proof Blocks and Parsons Problems}\label{sec:parsonsproofblocks}

% A Synopsis 
Parsons problems were first introduced by~\citet{parsons2006parson} as
``Parsons programming puzzles'' which is an engaging, interactive tool for
scaffolding the process of learning to write code. Since their introduction,
they have become one of the most well investigated tools in computing
education~\cite{du2020review, ericson2022parsons} showing a wide variety of
affordance for students. \citet{ericson2017solving} showed in a variety of
studies that practice with Parsons problems is just as effective in learning
programming as traditional code writing exercises, but takes students less time
to complete which leads to more efficient learning. Additionally,
\citet{ericson2018evaluating,ericson2019investigating} introduced and evaluated
\textit{Adaptive} Parsons problems which allow students to use a help button to
remove distractors and combine blocks---the intention being to simplify the
problem to match the students current skill level.

% Proof blocks
Proof blocks~\cite{poulsen2022proof} extend the idea of Parsons problems to
the domain of discrete math, specifically, (as the name suggests) mathematical
proofs. They offer a similar interactive, drag-and-drop interface where
students construct a proof by arranging blocks of a proof into a correct order.
This offers some unique challenges with respect to grading, such as the
inability to use execution based feedback as in Parsons problems and the fact
there may be a wide variety of correct orderings of the blocks. To address
this, \citet{poulsen2022proof} introduced a dependency graph grading approach
which allows question authors to embed the logical dependencies between blocks
of the proof presented to students in the block bank. Grading is then done by
accepting any valid topological sort of the DAG as a correct solution. To
further improve this grading approach and feedback mechanism,
\citet{poulsen2023efficient} introduced a partial credit mechanism which is
based on the edit distance between the students submission and the closest
correct solution.

\subsection{Distractors in Parsons Problems}

Distractors, which are incorrect blocks of code included in the problem
bank, have been a common feature of Parsons problems since their
inception~\cite{parsons2006parson}. Proof Blocks similarly include
distractors to highlight common proof errors. While findings regarding their
effectiveness are
mixed~\cite{harms2016distractors,ericson2023multi,poulsen2025measuring}, a
study by \citet{smith2024distractors} found positive results which demonstrated
through think-aloud interviews that distractors encourage students to engage
deeply with implementation details. However, significant work remains both in
identifying different varieties of distractors both in the domains of
programming and proofs and replicating these findings. The multigraph collapse
algorithm described in this paper introduces a novel perspective on
distractors. A student, by selecting one of the optional blocks, and thereby
venturing down one solution path, renders all other optional blocks associated
with other solution paths as distractors. This reconceptualization of
distractors, though not the primary focus of this work, opens up new avenues
for further research and highlights a potential design affordance of the
multigraph collapse algorithm.

\subsection{Worked Example Theory}\label{sec:workedexample}

Fundamentally, when Parsons problems are used in a formative context, they act
as a worked example. Worked examples offer both the opportunity to expose
students to an expert's solution and to scaffold the process of them
constructing that solution~\cite{barbieri2023meta,shen2009design}. Ideally,
this has both the effect of teaching students to follow the ``best practices''
being modeled by the expert's solution and doing so at a more manageable
cognitive load~\cite{sweller2006worked}. Worked example theory has been one of
the primary theoretical frameworks guiding the design and proposed utility of
both Parsons problems~\cite{ericson2022parsons} and Proof
blocks~\cite{poulsen2022proof}.

However, a core limitation of existing block-ordering implementations is that
it only supports a single correct solution path. Therefore, if instructors want
students to engage with multiple worked examples \textit{for the same} problem,
they must create multiple, distinct problems. By allowing optional blocks, and
therefore multiple distinct correct solutions, in a single block bank, we can
allow students to engage with multiple worked examples and explicitly adapt
existing submissions to form a new solution.

\section{System Implementation} \label{sec:implementation} The algorithms and
techniques detailed in this section were implemented in a fork of the
PrairieLearn block ordering plugin, where we build on top of the DAG grading
feature implemented by \citet{poulsen2023efficient}.

\subsection{Multigraph Collapse Algorithm} The aim of this algorithm is to take
a multigraph that allows multiple colored edges between the same pair of nodes,
then collapse it into all of its possible valid DAGs. In the context of an
order-blocks problem, colored edges represent an \emph{or} relationship between
different blocks and their one common dependency. Instructors can use these
multigraphs to model dependency relations with optional or alternative
requirements more easily than explicitly declaring every possible correct
answer.

\begin{algorithm}[!htb]
	\SetKwInOut{Input}{Input}
	\SetKwInOut{Output}{Output}

	\Input{
		\Desc{$M$}{The questions' multigraph.}\\
		\Desc{$F$}{Final node in the multigraph.}\\
	}
	\Output{
		\Desc{$CD$}{List of fully collapsed DAGs}
	}

	\Fn{\textsc{Collapse}($M$, $F$)}{
		$Collapsed\ Graphs\ (CD) \gets$ empty list\;
		$Partially\ Collapsed\ Graphs\ (PCGs) \gets$ empty queue\;

		\textsc{Enqueue}($PCGs$, M)\;

		\While{$PCG$ is not empty}{
			$G$ $\gets$ \textsc{Dequeue}($PCGs$)\;
			$(v, DAG) \gets$ \textsc{DFSuntil}($G$, $F$)\;

			\If{$v$ is NULL}{
				append $DAG$ to $CD$\;
			}

			\ForEach{Color in $G$}{
				\tcp{collapse operation}
				$PCG$ $\gets$ copy($G$)\;

				Remove all \emph{edges} on $v$ of \emph{PCG} \\
				\Indp
				with edges whose color is not $Color$\;
				\Indm

				\textsc{Enqueue}(PCGs, PCG)\;
			}

		}
		\Return{$CD$}\;
	}
	\caption{Collapse a multigraph into all valid DAGs}
	\label{alg:multigraph-collapse}
\end{algorithm}

For example, if line A can depend on either line B or line C (but not
necessarily both), this can be represented as two differently colored edges
from line A to both line B and line C. Grading becomes straightforward because
we can take the student's submission, iterate through the list of all possible
valid orderings generated by collapsing the multigraph, and return the minimum
edit distance across all possibilities for maximum partial credit, or full
credit if any edit distance is $0$.

The multigraph collapse algorithm takes a multigraph $M$ and the multigraph's
final node $F$ (the node that has only incoming edges, representing the final
block in the solution). Assuming the multigraph is acyclic, the algorithm
utilizes a modified depth-first search to traverse the dependency graph
backward from node $F$ until it either:

\begin{itemize}
	\item Completes the traversal (reaching all initial nodes with only incoming edges); or
	\item Encounters a node with multiple colored edges requiring a choice.
\end{itemize}

When the traversal completes without encountering colored edges, we have
successfully identified a valid DAG and add it to our list of collapsed DAGs
($CD$). When we encounter a node with any colored edges, we must split the
problem. For each color present at that node, we create a copy of the current
graph, remove all edges of other colors from that node, and continue
processing. This branching ensures we explore all possible combinations of
every choice. The full details of the algorithm are shown in
Algorithm~\ref{alg:multigraph-collapse}.

\begin{theorem}
	The algorithm's time complexity is $\bigO{(d \cdot (n + m))}$ where $n$ is
	the number of nodes, $m$ is the number of edges, and $d$ is the number of DAGs generated by
	collapsing all possible edge combinations.
\end{theorem}

\begin{proof}
	The algorithm generates an upper bound of $d$ distinct DAGs, where $d$ represents all
	possible combinations of color choices. For each of the generated DAG's we
	perform a traversal \bigO{(n + m)}. The total number of splits is bounded
	by $d-1$ (because we generate $d$ final DAGs). Therefore, the time
	complexity is $\bigO{(d \cdot (n + m))}$
\end{proof}

\begin{table*}[ht]
	\centering
	\renewcommand{\arraystretch}{1.2}
	\setlength{\tabcolsep}{10pt}
	\rowcolors{2}{white}{rowgray}
	\begin{tabular}{p{2cm}|p{8.25cm}|c|c|c}
		\rowcolor{steelblue}
		\textcolor{white}{Content}     &
		\textcolor{white}{Description} & \textcolor{white}{$n$
		(Blocks)}                      & \textcolor{white}{$m$ (Edges)}                       & \textcolor{white}{$d$ (DAGs)}          \\
		\hline
		Bash                           & Change the origin of a git repo.                     & 5                             & 5  & 2 \\
		Bash                           & Remove the final git tag                             & 12                            & 7  & 2 \\
		Bash                           & Commit and view the git log.                         & 5                             & 4  & 2 \\
		Bash                           & Commit and view the git log (alternate).             & 9                             & 4  & 2 \\
		\hline
		Python                         & Write a function to sum two numbers.                 &
		6                              & 8                                                    & 3                                      \\
		Python                         & Create \texttt{hello.txt} containing the phrase
		\texttt{hello}.                & 13                                                   & 10                            & 4      \\
		Python                         & Reverse words in a string.                           & 8                             & 12 & 8 \\
		Python                         & Read a file and output a formatted result.           & 8                             & 13 & 8 \\
		Python                         & Find all keys in a dictionary with a specific value. & 8                             & 9  & 4 \\
		Python                         & Filter names starting with
		A.                             & 8                                                    &
		9                              & 2                                                                                             \\
		\hline
		Proof                          & Prove that if $n$ is even, then $n+10$ is even       & 11                            & 13
		                               & 2                                                                                             \\
		Proof                          & Prove that the cardinality of two sets is equal      & 10                            &
		10                             & 2                                                                                             \\
	\end{tabular}

	\vspace{1mm}
	\caption{Analysis of the complexity of Bash and Python questions. $n$ = number of code blocks, $m$ = number of edges, $d$ = number of valid DAGs.}
	\label{tab:question-summary}

\end{table*}

We note that the number of DAGs $d = \prod_i c_i$ for all nodes $i$, where
$c_i$ is the number of distinct colors on node $i$'s outgoing edges, the
multigraph collapse algorithm upper bound of DAGs needing to be traversed is
exponential. In theory, this could prohibit its use in some questions, however,
in practice questions with many possible orderings would be impractical for
student learning. As shown in Table~\ref{tab:question-summary}, all questions
which have been written so far across multiple domains have $d \le 8$.

%Writing an effective question gets more difficult as the number of orderings
%increases. Every new split has to be kept track of and each new ordering has to
%ensure correctness. Attempting to write an effective question with only up to
%twenty orderings would be a very difficult task and one that prevents the
%algorithm from reaching a runtime that could reach unsuable. In a case where
%there are twenty plus orderings the question writer would more than likely
%generate invalid orderings that were unintended, rendering the question invalid.

\subsection{Depth First Search Until Algorithm} We separate the portion of our
algorithm which finds the next node where a collapse needs to happen into a
separate algorithm which we call DFS-Until. This algorithm is a basic
\emph{Depth-First Search} with a halting condition to stop searching and return
the reason with the nodes visited. This allows for situations where a graph may
need some modifying to shape it into a certain form, correct errors, or in this
case alert the multigraph collapse algorithm there are colored nodes. After
returning the position of the colored nodes we continue searching from the
halting position. The algorithm takes the halting condition as input along with
the graph and a starting node. The complexity of this algorithm is the same as
a standard \emph{DFS}, where a complexity is \bigO{($n+m$)} where $n$ is the
number of nodes and $m$ is the number of edges.

\begin{algorithm}[t]
	\Input{\\
		\noindent\hspace{1.5em}\Desc{$H$}{The halting condition for Depth-First Search.}\\
		\noindent\hspace{1.5em}\Desc{$G$}{The graph to be traversed.}\\
		\noindent\hspace{1.5em}\Desc{$N$}{Starting node for the Depth-First Search.}
	}
	\Output{\\
		\noindent\hspace{1.5em}\Desc{$v$}{Reason for halting.}\\
		\noindent\hspace{1.5em}\Desc{$DAG$}{Graph of traversed nodes and edges}
	}
	\Fn{\DFSuntil($H$, $G$, $N$)}{
		$S \leftarrow$ empty stack\;
		$DAG \leftarrow$ empty graph\;
		Push($N$, $S$) \tcp*{Put starting node on the stack}
		\While{$S$ is not empty}{
			$v \leftarrow$ Pop($S$)\;
			Add $v$ and all it's incoming edges to $DAG$\;
			\If{$v$ meets $H$}{
				\Return($v$, $DAG$) \tcp*{Halting condition met}
			}
			\ForEach{Incoming edge $(u,v)$ of $v$}{
				Push $u$ to $S$\;
			}
		}
		\Return(NULL, $DAG$)\;
	}
	\caption{Depth‑first search with halting condition}
	\label{alg:dfs-until}
\end{algorithm}

%\subsection{Named Paths} \label{sec:named-paths} The named paths feature was
%added after a question was written that needed further filtering to ensure that
%a valid DAG but invalid solution was not included in the final list of
%collapsed DAGs. The feature requires additional information to be specified by
%the writer, where if a path is named its edges must be completely included. The
%implementation details are beyond the scope of this paper and have been omitted
%from Algorithm~\ref{alg:multigraph-collapse}.

\section{Example Usage in Different Domains}

%\vspace{-3em}
\subsection{Discrete Math}

Like in programming, in writing proofs there are also often multiple methods to
successfully complete a proof. In the example Proof problems we have
demonstrated a few of these ways. For example, consider the following problem
that may appear in the first few weeks of a discrete math course. While
teaching students basic proof techniques, for example: ``Prove that if $n$ is
even, then $n + 10$ is even''. It can be useful to help students see that this
fact can be proven either directly, by assuming $n$ is even and then showing
$n+10$ is even, or by contradiction, by assuming that $n+10$ is odd and then
showing a contradiction.

Later in the course, discrete math courses usually cover set cardinality, and
have students write proofs that two sets have the same cardinality. When
proving that two sets $A$ and $B$ have equal cardinality, one can construct and
injection $f: A \rightarrow B$ and an injection $g: B \rightarrow A$ and use
the Schr\"oder–Bernstein theorem to conclude that the sets have equal
cardinality. However, a corollary to the Schr\"oder–Bernstein theorem is that
one or both of the injections could be replaced by a surjection in the opposite
direction. So the proof could equivalently proceed by defining injection $f: A
	\rightarrow B$ and a surjection $g: A \rightarrow B$. This is another
opportunity where the multigraph collapse algorithm allows us to give students
problems that permit them to explore more varied solutions to a problem, as the
students can be presented with the blocks to complete the proof either way.

\begin{tcolorbox}[
		colback=gray!10, % very light gray background
		colframe=gray!10, % no visible border
		boxrule=0pt, % no border thickness
		arc=2pt, % slightly rounded corners
		left=4pt, % small padding on the left
		right=4pt, % small padding on the right
		top=2pt, % minimal top padding
		bottom=2pt, % minimal bottom padding
		fontupper=\small, % smaller text to de-emphasize
		before skip=6pt, % small space before box
		after skip=6pt % small space after box
	]
	\textbf{Key Use Cases:} In discrete math, it is useful to show students that proofs can be
	constructed multiple ways, such as directly or by contradiction.
\end{tcolorbox}

\subsection{Introductory Programming}

Block-ordering problems in the context of an introductory programming
course, often referred to as Parsons problems~\cite{parsons2006parson}, have had
a variety of grading approaches, including: 1) execution-based grading, where
the student submission is executed and its output is compared against unit
tests, and 2) order-based grading, where, much like in the context of discrete
math, the ordering of the blocks students submit is compared to one or more
possible orderings~\cite{helminen2013students, wu2023using, wu2024evaluating}.
Each of these comes with its own unique feedback mechanisms, with
execution-based grading often providing feedback through test case results and
order-based grading often using a top-down approach to feedback, where the
student is shown which blocks are correctly placed until a block that is
incorrectly placed is reached. The use of block-ordering problems in the domain
of programming is unique in that the use of optional blocks, and therefore,
support for multiple correct solutions, is enabled natively through
execution-based grading. However, when using a simple ordering comparison where
only one potential ordering is possible, or a DAG-based grading
approach~\cite{poulsen2023efficient} where multiple correct orderings are
possible, the multigraph collapse algorithm allows instructors to create
problems that allow for multiple, distinct correct solutions while retaining
the simplicity and specificity of ordering-based feedback.

\subsubsection{Instructor Experiences:} This optional block feature was used by
the second author of this paper in a large introductory programming course,
both on homework and exams. The support for optional blocks primarily enabled
the creation of questions that allowed students to submit solutions that
performed operations either through iteration or through function chaining. For
example, the instructor constructed one problem where students were asked to
read names from a file where each name appeared on a new line, and return them
as a comma separated string. The use of optional blocks allowed for two valid
approaches. One, where the student iterated through the file using
\inlinecode{readline()} and collected each name by appending it to a list.
Two, it allowed for another solution where the entire file was read using
\inlinecode{read()} and the list of names was constructed with
\inlinecode{names = file\_content.strip().split("\textbackslash{}n")}.

\begin{tcolorbox}[
		colback=gray!10, % very light gray background
		colframe=gray!10, % no visible border
		boxrule=0pt, % no border thickness
		arc=2pt, % slightly rounded corners
		left=4pt, % small padding on the left
		right=4pt, % small padding on the right
		top=2pt, % minimal top padding
		bottom=2pt, % minimal bottom padding
		fontupper=\small, % smaller text to de-emphasize
		before skip=6pt, % small space before box
		after skip=6pt % small space after box
	]
	\textbf{Key Use Cases:} In the context of introductory programming courses,
	the multigraph grading approach supports having students solve the same
	problem with different paradigms side-by-side (e.g., iterative vs.\
	functional, iterative vs.\ recursive).
\end{tcolorbox}

\subsection{Shell Commands and The Shell Tutor ITS}

% Admittedly a bit long winded and repeats a lot of what was noted before. I'm
% sure we can trim this down significantly.

One of the key motivators of this work was to enable the rich support of
Parsons problems for sequences of shell commands particularly, to assess
student knowledge of using shell commands to perform complex tasks learned
using the Shell Tutor, which is an intelligent tutoring system which teaches
the Unix command shell \cite{winder2024shelltutor}. Singular shell commands, or
short sequences of commands, can accomplish the same task. To delete a
directory in the Unix command shell, one can use either the command \inlinecode{rm -r
	directory} or the sequence of commands \inlinecode{rm directory/*} and \inlinecode{rmdir
	directory} to delete the directory and files. The Shell Tutor would accept
either of these sequences as a valid solution for this task, but would not
accept both. There is no reason to attempt to delete the directory twice. When
assessing student knowledge of the command shell using Parsons problems,
it allows a student to construct both possible solutions--just as they are able
to in the intelligent tutor--assists in emulating the environment in which they
learned, as well as the real world usage of the command shell.

% This subsubsection could potentially be removed entirely
\subsubsection{Instructor Experiences:} The third author has utilized these
optional dependency Parsons problems to assess student command shell and Git
knowledge in a Fundamentals of Software Engineering course. One example problem
assesses student understanding of Git remotes by manipulating a local
repository's remote: changing the \inlinecode{origin} remote to a new URL and
retaining the original remote URL under a new name. One solution sequence is
``rename the existing remote, add a new remote,'' while the second solution is
``delete the old remote, add the new remote, add the old remote again with a
different name.'' These two sequences of operations achieve the same result,
however proper construction of these sequences requires a higher level of
understanding of Git than rote command repetition. In practice, the third
author has noticed anecdotal evidence that this question is a great
discriminator when assessing student understanding of Git.

% I'm adding the note that I have seen anecdotal evidence for this being a
% discriminator because I haven't run the most recent numbers. I do not want to
% make a concrete claim here that I don't have proper evidence to backup at
% this time :) 

\begin{tcolorbox}[
		colback=gray!10, % very light gray background
		colframe=gray!10, % no visible border
		boxrule=0pt, % no border thickness
		arc=2pt, % slightly rounded corners
		left=4pt, % small padding on the left
		right=4pt, % small padding on the right
		top=2pt, % minimal top padding
		bottom=2pt, % minimal bottom padding
		fontupper=\small, % smaller text to de-emphasize
		before skip=6pt, % small space before box
		after skip=6pt % small space after box
	]
	\textbf{Key Use Cases:}  Problems can be constructed with the intention of
	highlighting both approaches to achieving the same outcome while
	ensuring the operations used to do so does not interfere with each other and
	ensures redundant operations are avoided.
\end{tcolorbox}

\section{Formative and Summative Settings}

To highlight the utility of this system, we present and discuss several examples
of potential use cases and propose future work evaluating block-ordering
problems with optional blocks in both formative and summative settings.

% Key Idea: Desirable difficulties and scaffolding
\subsection{Multiple Solutions in Practice Settings}

% Prior Work
A key benefit of block-ordering problems is their ability to expose students to
expert solutions as worked examples~\cite{ericson2022parsons}. In prior work,
a commonly cited affordance of block-ordering problems is that students are
more efficient at solving these questions in comparison to an equivalent code writing
question while still maintaining similar learning gains under both practice
conditions~\cite{ericson2017solving, ericson2023multi}. However, as noted by
\citet{haynes2022impact, haynes2021problem}, many students struggle to solve
block-ordering problems that contain what they term ``uncommon solutions'',
which are solutions that novice students are unlikely to produce. This raises
the question of how learning efficiency should be balanced against potentially
desirable difficulties~\cite{bjork2020desirable, bjork2011making}, such as
exposing students to a variety of solutions to a problem that they may not have
encountered before.

% Optional line block ordering problems to the rescue!
The multigraph collapse algorithm may help address this challenge by enabling
instructors to create block-ordering problems with multiple correct solution
paths and, if desired, require students to \textit{discover all correct
	solutions} to the problem. In this scenario, it is likely that a student
solving such a problem would first discover the solution path that most closely
matches their preconception of the correct solution. From there, they can
explore what changes to their initial solution would yield a different, but
still correct, solution to the problem. In such a scenario, their initial
solution would act as a scaffolding which would allow them to explore the
problem space without being potentially overwhelmed by having to construct an
uncommon or potentially counterintuitive solution from scratch.

% Proposed future work
\subsubsection{Future Work} Future work should consider the development of
interfaces for supporting optional blocks in block-ordering problems with
a focus on mitigating the extraneous cognitive load that may be caused by
providing the sometimes large number of blocks needed to support a variety of
solution paths. Additionally, future work should consider the learning benefits
of allowing students to explore multiple correct solutions in practice settings
and the degree to which this exploration supports a more robust understanding
of both the problem they are solving and the range of approaches available to
solve said problem.

\subsection{Allowing Many Correct Solutions in Exams}

% Key Idea: 
In the context of exams, the work of ~\citet{haynes2022impact, haynes2021problem}
presents a critical consideration for the use of Parsons problems with only one
available solution in exams: do students whose preconception of the correct
solution match the allowed solution path gain an unfair advantage over students
who have a different preconception of the correct solution? Certainly, prior work
would suggest that students whose preconceived solution matches the available
solution path would be able to solve the problem faster, a key consideration in
the time constrained environment of exams.

\subsubsection{Future Work:} Future work should consider both the
differential advantage students gain from having a preconceived solution that
matches the available solution path in an exam setting and the degree to which
allowing for multiple correct solutions mitigates this advantage.

\section{Conclusion} We have presented a novel approach to support optional
blocks in Proof Blocks DAG grading software \cite{poulsen2022proof} by using a
multigraph-based dependency model. This tool not only allows instructors more
freedom to write more complex and interesting questions, but also helps
illustrate the principle that multiple valid solutions can exist for a given
problem.

\bibliographystyle{ACM-Reference-Format}
\bibliography{acmart}
\balance

\end{document}